\documentclass[runningheads]{llncs}
\usepackage{geometry}
\usepackage[numbers]{natbib}
\usepackage{doc}
\usepackage{url}
\usepackage{graphicx}
\usepackage{epstopdf}
\usepackage{hyperref}
\usepackage{mathtools}
\usepackage{algorithm}
\usepackage{algorithmic}
\usepackage{nccmath}
\usepackage{amsmath}
\title{Maximum-Profit Routing Problem with Multiple Vehicles per Site}
\author{Bogdan Armaselu}
\institute{barmaselub@gmail.com}
\date{}

\begin{document}

\maketitle
\begin{center}

\end{center}

\abstract{
We consider the Maximum-Profit Routing Problem (MPRP), a variant of pick-up routing problem introduced in \cite{Armaselu-arXiv-2016, Armaselu-PETRA}, in which the goal is to maximize total profit.
The original MPRP restricts vehicles to visit a site at most once.
In this paper, we consider extensions of MPRP, in which a site may be visited by a vehicle multiple times.
Specifically, we consider two versions: 
one in which the quantity to be picked up at each site is constant in time (MPRP-M),
and one with time-varying supplied quantities, which increase linearly in time (MPRP-VS).
For each of these versions, we come up with a constant-factor polynomial-time approximaltion scheme.
}


\large

\section{\Large{Introduction}}
\label{s:intro}

Consider a set of $n$ points in the plane, called \textit{sites}.
Each site $S_i$ supplies a quantity $q_i$ of a certain unit-priced product that needs to be collected
and has an operating time window of $[s_i, e_i]$, where $0 \leq s_i \leq e_i \leq T$, for some constant $T > 0$, $\forall{i = 1, \dots, n}$.
The distance $d_{i,j}$ between sites $s_i$ and $s_j$ is the euclidean distance.
We are also given a fleet of $m$ vehicles $V = \{v_1, \dots, v_m\}$, each having the same capacity $Q$.
All vehicles are assumed to travel at unit speed and have unit fuel consumption per distance travelled.
Only one vehicle may visit any given site and all vehicles must start and end their tour at the same given depot $D(x_D, y_D)$.
The goal is to compute, for every vehicle $v_k$, a route $r_k$ that collects a quantity $q^k_i$ of the good such that the total profit of all routes is maximized, 
where the profit of a route $r_k$ is the total quantity collected $\sum_{s_i \in r_k}{q^k_i}$ (called \textit{reward}) minus the total distance travelled via $r_k$ (called \textit{costs}).
We call this problem the Maximum Profit pick-up Routing Problem (MPRP).

\hspace{5mm}MPRP was introduced in \cite{Armaselu-arXiv-2016, Armaselu-PETRA}, which focuses on the case with fixed quantities and only one vehicle per site.
In this paper, we consider the case where multiple vehicles may visit one site.
We first study the MPRP problem with Multiple Vehicles Per Site and fixed quantities (MPRP-M),
in which $q_i$'s are constant in time.
We also consider the MPRP problem with Multiple Vehicles Per Site and Time-Variable Quantities supplied (MPRP-MVS),
in which $q_i(t)$'s are 0 for $t \notin [s_i, e_i]$ and linearly increasing in time for $t \in [s_i, e_i]$.
Compared to the single vehicle versions, this adds the availability constraint that $\sum_{j, u \leq t}{q_{iju}} \leq q_i(t), \forall{i, t}$.
However, a vehicle may not visit a site more than once.

\hspace{5mm}MPRP is known to be strongly NP-hard \cite{Armaselu-PETRA}.
By extensions, all the variants proposed above are also strongly NP-hard. 
That is, they have no algorithm running in time polynomial in the value of the input unless P = NP.

\hspace{5mm}Applications of the results presented in this paper include public transportation in various settings, such as cities \cite{Armaselu-PETRA}, inter-city railway transportation, domestic or international flight, etc.
Another important application is in for-profit waste pick-up, in which the reward is proportional to the amount of waste collected.

\subsection{\large{Related work}}
\label{ss:related-work}

General graph TSP is proven to be hard to approximate within any constant factor \cite{Sahni}.
However, TSP on special graphs, such as complete graphs with metric distances ane euclidean distances, do have approximation algorithms.
For instance, general metric TSP has a $O(n^3)$ time $1.5$-approximation algorithm by Christofides \cite{Christofides}.
Euclidean metric TSP even has a PTAS achieving $1 + \epsilon$ performance ratio via Arora's $O(n (\log n)^{O(1/\epsilon)})$ time algorithm \cite{Arora}.
On the other hand, metric TSP was shown to be APX-complete by Papadimitriou et. al \cite{Papadimitriou} 
and, to date, the best known lower bound is $123/122$ by Karpinski et. al \cite{Karpinski}.

\hspace{5mm}Various generalizations of TSP have also been considered.
Bansal et. al introduced the Deadline-TSP and the Time Window-TSP problems \cite{Bansal}.
In Deadline-TSP, every site $S_i$ has a deadline $D_i$, 
while in Time Window-TSP, every site $S_i$ has a time window $[s_i, e_i]$.
Both of these problems associate every site $S_i$ with a reward $q_i$ and ask for a single tour that maximized the total reward collected.
For Deadline-TSP, they provide an $O(\log n)$-approximation algorithm, while for Time Window-TSP they describe an $O(\log^2 n)$-approximation.

\hspace{5mm}There are also variants involving multiple vehicles.
Fisher introduced the Vehicle Routing Problem with capacity constraint (VRP), in which a fleet of $m$ vehicles is given, and each customer $S_i$ has a demand $q_i$ of a certain product \cite{Fisher94}.
Their solution uses an iterative lagrangian relaxation of the constraints.
Later, Fisher et. al introduced the Vehicle Routing Problem with Time Windows and capacity constraints (VRPTW),
 in which customers have a demand $q_i$ of the product, as well as operating time windows, and the $m$ vehicles have non-uniform capacity constraints \cite{Fisher95}.
They solve this problem using a linear programming approach.

\hspace{5mm}Golden et. al studied the Capacitated Arc Routing Problem (CARP) with uniform vehicle capacity 
where edges, rather than vertices, have customer demands, and the goal is to minimize the travelled distance subject to meeting all demands.
They prove that CARP can be approximated within a constant factor when the triangle inequality is satisfied \cite{Golden}.
Later, van Bevern extended the result to general undirected graphs \cite{Bevern}.

\hspace{5mm}Wang et. al considered the Multi-Depot Vehicle Routing Problem (MDVRPTW) with Timw Windows and Multi-type Vehicle Number Limits \cite{Wang}.
The main difference to our problem is that the goal is to minimize the number of vehicles used (if feasible) or maximize the number of visited customers (if infeasible). 
Although they solve the problem using a genetic algorithm approach, their algorithm is iterative and has no performance guarantee with respect to the approximation factor.

\hspace{5mm}Versions of vehicle routing with multiple vehicles per customer have also been considered.
In \cite{Drexl}, Drexl does a comprehensive survey of variations of vehicle routing with multiple synchronization constraints (VRPMS),
e.g. in which a customer may be visited by two vehicles (VRPTT) \cite{Bredstrom}, or by 3 or more vehicles \cite{Burckert}.

\hspace{5mm}Armaselu and Daescu solved the fixed-supply, single vehicle version of MPRP \cite{Armaselu-arXiv-2016, Armaselu-PETRA}.
They provide two APXs for MPRP, both running in $O(n^{11})$ time, assuming $q_i(l_i) \leq \alpha q_j(l_j), \forall{i, j = 1, \dots, n}$, for some constant $\alpha > 1$.
The best performance ratio, $15 \log T$, among the two algorithms, is achieved by an approach involving well-separated pair decompositions.

%

\subsection{\large{Our contributions}}
\label{ss:structure}

We give a constant-factor approximation algorithm for MPRP-M and MPRP-MVS, which uses clever reductions to some known problems.
Specifically, we solve MPRP-M in $O(n^{11})$ time for an approximation ratio of $\simeq 44 \log T$ and 
MPRP-MVS in $O((\frac{n}{\epsilon})^{11})$ time within an approximation ratio of $44 \log T (1 + \epsilon)(1 + \frac{1}{1 + \sqrt{m}})^2$.

\hspace{5mm}The rest of the paper is structured as follows.
In Section \ref{s:mprp-m}, we describe our MPRP-M algorthm, and then, in Section \ref{s:mprp-mvs}, we describe our solution to MPRP-MVS.
Finally, in Section \ref{s:conclusion}, we conclude and list some future directions.

\section{\Large{MPRP-M}}
\label{s:mprp-m}

We come up with a reduction from MPRP-M to MPRP. 
%
%
%

\hspace{5mm}Let $d(0, i) = d(i, 0)$ be the distance from the depot to site $S_i, \forall{i: 1 \leq i \leq n}$.

\hspace{5mm}An instance $I$ of MPRP (or MPRP-M) can be described as a set $I = (S = \{S_i: 1 \leq i \leq n\}, D = (x_0, y_0), d = \{d(i, j): 0 \leq i < j \leq n\}, T, Q, m)$,
where $S_i = (x_i, y_i, s_i, e_i, q_i)$.

\hspace{5mm}A solution $Sol$ to $I$ is a set $Sol = (R, P)$.
Here $R$ is a routing function assigning, to every pair of sites or depot $(i, j), 0 \leq i, j \leq n$, a number $k(i, j): 0 \leq k(i, j) \leq m$, 
denoting the index of the vehicle operating the path from $S_i$ to $S_j$, or 0 if no such vehicle exists.
$P$ is a pickup function assigning, to every site $S_i$ on vehicle $v_k$'s route, a quantity $q_{k,i} \leq q_i, Q$ to be picked up by $v_k$.

\hspace{5mm}A solution $Sol'$ to $I'$ is a routing function with the same properties as the function $R$ in$Sol$.

\hspace{5mm}Denote by $MPRP$ (resp., $MPRP-M$) the set of MPRP (resp., MPRP-M) instances.

\hspace{5mm}Starting from $I = (\{(x_i, y_i, s_i, e_i, q_i): 1 \leq i \leq n\}, (x_0, y_0), \{d_{i,j}: 0 \leq i < j \leq n\}, T, Q, m) \in MPRP-M$,
let $I' = (\{(x_i, y_i, s_i, e_i, q_i): 1 \leq i \leq n\}, (x_0, y_0), \{d_{i,j}: 0 \leq i < j \leq n\}, T, Q, m) \in MPRP$.
That is, $I'$ has \textit{the same nodes and distances} as $I$.
We solve $I'$ using the algorithm in \cite{Armaselu-PETRA} and denote by $Sol'$ the solution.

\hspace{5mm}We transform the solution $Sol'$ to an MPRP instance $I'$, into a solution $Sol$ to an MPRP-M instance $I$, as follows.
For every vehicle $v_k$ of $I'$, let $Q_k$ be the total quantity collected by $v_k$ from its assigned tour $\tau'_k \in Sol'$.
It may be possible that some sites in $\tau'_k$ need to be visited multiple times in $I$, producing a new tour $\tau_k \in Sol$.
To figure out which ones need to, we look at tuples of sites $(S_u, S_{u'}, S_i, S_{i'}, S_j, S_{j'}) \in S^2 \times (S \cup D)^4$ in $I'$ where 

\begin{equation}
\begin{cases}
S_u \in \tau'_k\\
S_{u'} \in \tau'_{k'}\\
S_i, S_j$ consecutive in $\tau'_k\\
S_{i'}, S_{j'}$ consecutive in $\tau'_{k'}
\end{cases}
\end{equation}

\hspace{5mm}The goal is to maximize the total amount picked up from $\tau'_k$ and $\tau'_{k'} \cup \{S_u\}$.
Denote by $q_{u', k'}$ the quanity assigned to be picked up by $v_{k'}$ from $S_{u'}$,
by $q_{u', k}$ the quanity assigned to be picked up by $v_{k}$ from $S_{u'}$,
and by $q_{u, k}$ the quanity assigned to be picked up by $v_{k}$ from $S_{u}$.
We have

\begin{equation}
\begin{cases}
Q_k + q_{u, k} + q_{u', k} - q_{u'} \leq Q\\
Q_{k'} + q_{u', k'} \leq Q\\
q_{u', k} + q_{u', k'} \leq q_{u'}\\
q_{u, k} \leq q_u
\end{cases}
\end{equation}

\hspace{5mm}and our goal is to maximize $q_{u', k} + q_{u', k'} + q_{u, k}$.
To do that, we set $x = q_{u, k}, y = q_{u', k}, z = q_{u', k'}$, and we come to the following linear program.

\begin{center} 
Minimize $x + y + z$ s.t.
\end{center}

\begin{equation}
\begin{cases}
x \leq q_u\\
y + z \leq q_{u'}\\
z \leq Q - Q_{k'}
\end{cases}
\end{equation}

\hspace{5mm}After solving the linear program, we check the additional constraint $x + y \leq Q - Q_k + q_{u'}$.
If it is satisfied, we assign $v_{k'}$ to pickup $z^*$ from $S_{u'}$, and $v_k$ to pickup $x^*$ from $S_u$ and $y^*$ from $S_{u'}$ in $Sol$,
where $(x^*, y^*, z^*)$ is a solution to the linear program.
In order for $Sol$ to achieve a better profit than $Sol'$ through this re-assignment, the following need to hold.
Let $S_u < S_v$ denote that $S_u$ is visited before $S_v$ in $\tau'_{k'}$.

\begin{equation}
\begin{cases}
x^* + y^* + z^* - q_{u'} > d(i', i) + d(i, i'') - d(i', i'')\\
d(i', i) + d(i, i'') - d(i', i'') \leq \min\{t_{k', v'} - s_{v'}: S_{v'} < S_{i'}\} + \min\{e_{w'} - t_{k', w'}: S_{j'} < S_{w'}\}
\end{cases}
\end{equation}

where $t_{k, i}$ is the time when $v_k$ visits $S_i$.
That is, insertion of $S_i$ into $\tau_{k'}$ in $Sol$ does not introduce time window violations

\hspace{5mm}See Figure 2 for an illustration of this re-assignment.

\begin{figure}
\label{fig:reduction-option1}
\begin{center}
	\includegraphics[scale=0.4]{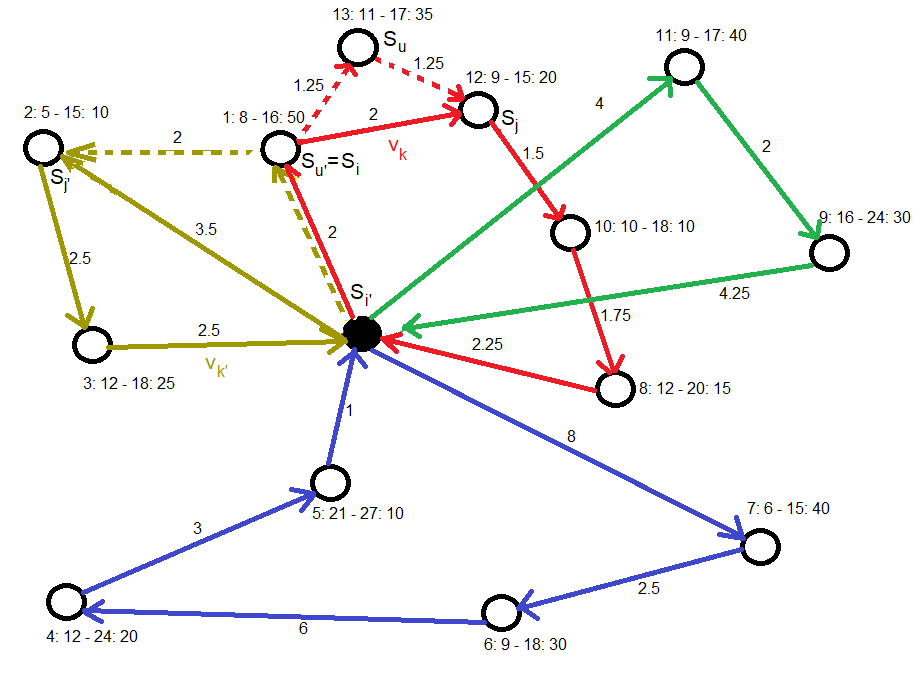}
	\caption{Tours of $Sol'$ are depicted in solid lines and tours of $Sol$ are depicted using dashed lines. 
	$v_k$ is assigned to pickup 35 from $S_u$ and 20 from $S_{u'}$, while $v_{k'}$ is assigned to pickup 35 from $S_{u'}$.
	Since time windows are not violated, $Sol$ is more profitable than $Sol'$.}
\end{center}`
\end{figure}

\hspace{5mm}For each vehicle $v_k$ and site $S_{u'} \in \tau'_k$, 
we look at all possible tuples $(S_u, S_i, S_{i'}, S_j, S_{j'}) \in S \times (S \cup D)^4$ satisfying the abovementioned additional constraints and solve the linear program to find the optimal quantities to pickup, 
then select the tuples that maximize the profit gain after re-assignment, and then perform the re-assignment.

\hspace{5mm}One arising concern is that the order in which the pairs $(v_k, S_{u'})$ are selected produces different results.
However, while this may happen, it may not impact the optmiality of the re-assignment by more than a constant factor, as we shall see below.

\begin{lemma}
Let $O$ be an ordering of the pairs $(v_k, S_{u'})$ and let $\Delta(k, u')$ be the gain obtained through re-assignment when selecting pair $(v_k, S_{u'})$ (or 0, if re-assignment is infeasible).
Then $\sum_{k, S_{u'} \in \tau'_k} \Delta(k, u') = \Delta \in [C, 4 C]$, where $C$ does not depend on $O$.
\end{lemma}

\begin{proof}
From (4), it follows that
\[
\Delta(k, u') = \max\{0, x^* + y^* + z^* - q_{u'} - d(i', i) - d(i, i'') + d(i', i'')\}
\]
\[
\geq \max\{0, x^* + y^* + z^* - q_{u'}\}
\]
\[
 + \max\{-t_{k', v'} + s_{v'}: S_{v'} < S_{i'}\}
\]
\[
 + \max\{-e_{w'} + t_{k', w'}: S_{w'} > S_{j'}\}
\tag{5} \label{eq5}
\]

where $S_i < S_j$ indicates that $S_i$ is before $S_j$ in $\tau'_{k'}$.

\bigskip

\hspace{5mm}Since $x^* = q_{u, k}, y^* = q_{u', k}, z^* = q_{u', k'}$, we get 
\[
\Delta = \sum_{k, S_{u'} \in \tau'_k} \Delta(k, u')
\]
\[
\geq \sum_{k, S_{u'} \in \tau'_k} \max\{0, q_{u, k} + q_{u', k} + q_{u', k'} - q_{u'}\}
\]
\[ 
+ \sum_{k, S_{u'} \in \tau'_k} \max\{-t_{k', v'} + s_{v'}: S_{v'} < S_{i'}\}
\]
\[ 
+ \sum_{k, S_{u'} \in \tau'_k} \max\{-e_{w'} + t_{k', w'}: S_{w'} > S_{j'}\}
\tag{6} \label{eq6}
\]

\hspace{5mm}We have 
\[
\sum_{k, S_{u'} \in \tau'_k} q_{u', k} = \sum_{k, S_{u'} \in \tau'_k} q_{u'} = \sum_{k, S_i \in \tau'_k} q_i,
\tag{7} \label{eq7}
\]
\[
\sum_{k, S_{u'} \in \tau'_k} \max\{-t_{k', v'} + s_{v'}: S_{v'} < S_{i'}\} \geq \sum_k (s_i - t_{k, i}: i \in \tau_k),
\tag{8} \label{eq8}
\]
\[
\sum_{k, S_{u'} \in \tau'_k} \max\{-e_{w'} + t_{k', w'}: S_{w'} > S_{j'}\} \geq \sum_k (t_{k, i} - e_i: i \in \tau_k).
\tag{9} \label{eq9}
\]

\bigskip

\hspace{5mm}Let 

$Q^* = \sum_{\tau'_k, S_i \in \tau'_k} q_i, Q^s = \sum_k (s_i - t_{k, i}: i \in \tau_k), Q^e = \sum_k (t_{k, i} - e_i: i \in \tau_k)$.

Since $\sum_{k, S_{u'} \in \tau'_k} \max\{0, q_{u, k} + q_{u', k}\} \geq Q^*$, it follows that

\[
\Delta \geq Q^* + Q^s+ Q^e.
\tag{10} \label{eq10}
\]

Moreover, note that $Q^s + Q^e = \sum_k(s_i - t_{k, i} + t_{k, i} - e_i) = \sum_k(s_i - e_i)$.

That is, $\Delta \geq C = Q^* - \sum_k(e_i - s_i)$, where $C$ does not depend on $O$.

\bigskip

\hspace{5mm}Note that 
\[
x^* + y^* + z^* - q_{u'} \leq 2 (d(i', i) - d(i, i'') + d(i', i'')), \forall{u', i, i', i''},
\tag{11} \label{eq11}
\]

since otherwise one would be able to create an optimal metric-space MST-based tour of cost $> 2 C$, which is a contradiction of the result in \cite{Christofides}.

\hspace{5mm}Thus,
\[
\Delta \leq 2 (\sum_{k, S_{u'} \in \tau'_k} \max\{0, q_{u, k} + q_{u', k} + q_{u', k'} - q_{u'}\}
\]
\[ 
+ \sum_{k, S_{u'} \in \tau'_k} \max\{-t_{k', v'} + s_{v'}: S_{v'} < S_{i'}\}
\]
\[ 
+ \sum_{k, S_{u'} \in \tau'_k} \max\{-e_{w'} + t_{k', w'}: S_{w'} > S_{j'}\}).
\tag{12} \label{eq12}
\]

\hspace{5mm}Since 
\[
\max\{-t_{k', v'} + s_{v'}: S_{v'} < S_{i'}\} + \max\{-e_{w'} + t_{k', w'}: S_{w'} > S_{j'}\}
\]
\[
\leq 2 \max{s_{v'} - e_{w'}} \leq 2 T,
\]

for sufficiently large $\frac{\min{q_i}}{2 T}$, which is reasonable since $T$ is constant, we get
\[
\Delta \leq 4 \sum_{k, S_{u'} \in \tau'_k} \max\{0, q_{u, k} + q_{u', k} + q_{u', k'} - q_{u'}\}
\]
\[
\leq 4 (Q^* + Q^s+ Q^e).
\tag{13} \label{eq13}
\]

\hspace{5mm}That is, $\Delta \leq 4 C$.
\end{proof}

Thus, we first run a MPRP solver to obtain a routing $\tau'$ and then select the pairs $(v_k, S_{u'})$ in the order given by each $\tau'_k$ for every $k$.
For each of the $O(m n)$ pairs, we inspect all possible $O(n^5)$ tuples in $O(n)$ time per tuple, which is the time required to verify constraint 2.
That is, our reduction is done in $O(m n^7)$ time.

\hspace{5mm}As for the correctness of the reduction, suppose there exists a routing $\tau''$ for $I'$ yielding a profit increased by more than a factor of 4 compared to $\tau'$.
In order to do that, $\tau''$ must pickup a quantity $q''_{u'} >4  q_{u'}$ from some $S_{u'}$.
However, this implies that either $S_{u'}$ is not re-assigned in $\tau$ and thus gives an increased profit for $I'$, 
or is re-assigned to one of $S_u, S_i, S_{i'}, S_j, S_{j'}$ and thus gives an increased profit for one of these sites.
Both options lead to a contradiction.
Now suppose there exists a routing $\tau''$ for $I$ yielding an increased profit compared to $\tau$.
In order for that to happen, either some site $S_{u'}$ that was never added or involved in a re-assignment picks up a quantity $q''_{u', k} > q_{u', k}$ ,
implying an increased profit for $q''_{u'} > 4 q'_{u'}$  for $S_{u'}$ in $\tau'$,
or some site $S_{u'}$ involved in a re-assignment with $(S_u, S_i, S_{i'}, S_j, S_{j'})$ picks up a quantity $q''_{u', k} > 4 q_{u', k}$ ,
implying a quantity was picked up from some site (wlog assumed to be $S_u$) in $\tau''$ which is more than 4 times the one picked from the samee site in $\tau'$.
In both cases, a contradiction follows due to Lemma 1.
Hence, we have reduced MPRP-M to MPRP within an approximation ratio of 4 in $O(m n^7)$ time.

\hspace{5mm}By using the $\simeq 11 \log T$-approximation for MPRP in \cite{Armaselu-PETRA} for $O(n^{11})$ time, we thus solve MPRP-M in $O(n^{11})$ time.

\hspace{5mm}We have proved the following result.

\begin{theorem}
MPRP-M can be solved in $O(n^{11})$ time for an approximation ratio of $\simeq 44 \log T$.
\end{theorem}

%
%
%
%
%
%
%

\section{\Large{MPRP-MVS}}
\label{s:mprp-mvs}

In MPRP-MVS, quantities $q_i$ supplied at sites vary linearly as a function of time, i.e. $q_i(t) = q_i(e_i) \frac{t - s_i}{e_i - s_i}$. 

\hspace{5mm}We adapt the algorithm in \cite{Armaselu-arXiv-2020} for solving MPRP-VS, to work in our case, 
by using a similar reduction as the one from MPRP-M to MPRP described in the previous section, to reduce MPRP-MVS to MPRP-VS.

\hspace{5mm}Given an instance $I = (S = \{S_i: 1 \leq i \leq n\}, D = (x_0, y_0), d = \{d(i, j): 0 \leq i < j \leq n\}, T, Q, m)$ of MPRP-MVS, 
we apply the reduction in \cite{Armaselu-arXiv-2020} to obtain an instance $I'' \in MPRP-M$, and then solve MPRP-M using the algorithm in Section \ref{s:mprp-m}.
Again, this is not straightforward since the approach in \cite{Armaselu-arXiv-2020} depends on single vehicle assignment which is not the case here.
Thus, we need a more insightful construction of $I''$ than directly applying the algorithm in \cite{Armaselu-arXiv-2020}.

\hspace{5mm}For each site $S_i$, we split $[s_i, e_i]$ into $N = 1 + \frac{\ln(1/\alpha)}{\epsilon}$ intervals, for some $\epsilon > 0$, where $\alpha > 1$ is the smallest constant such that $q_i(e_i) \leq \alpha q_j(e_j), \forall{i, j}$.
Denote these intervals by $L_{i, l}: 1 \leq l \leq N$, i.e. $L_{i,l} = [s_i + \frac{e_i - s_i}{(1+\epsilon)^{N-l+1}}, s_i + \frac{e_i - s_i}{(1+\epsilon)^{N-l}}]$.
After performing this split, we construct a set of sites $S'$ where, for each original site $S_i$, $S'$ contains $N$ sites $S''_{i, \tau}, 1 \leq \tau \leq N$.
To each newly constructed site $S''_{i, l}$, we assign a quantity $q'_{i, l} = q_i(e_i) \frac{l - 0.5}{N - 1}$ and a time window $I_{i,l}$.
We then run the MPRP-M Option 1 algorithm described in Section \ref{s:mprp-m} on the transformed instance $I'' = (S'', D, d, T, Q, m)$ 
(note that $I'' \in MPRP-M$ since the sites in $S''$ now supply constant quantities).

\hspace{5mm}Denote by $\tau''$ the routing obtained by running the MPRP-M algorithm on $I''$.
We now analyze the properties of $\tau''$ in order to put a bound on its performance ratio.

\hspace{5mm}We know from \cite{Armaselu-arXiv-2020} that an optimal MPRP algorithm run on $S'$ may collect a quantity at least $\frac{1}{1 + \epsilon}$ as much as an optimal MPRP-VS algorithm run on $S$.
Since for $m = 1$, an instance of MPRP-M also belongs to MPRP, and an instance of MPRP-MVS is also an instance of MPRP-VS, we get the following.

\begin{lemma}
For $m = 1$, an optimal MPRP-M algorithm when run on $S''$ may collect a quantity at least $\frac{1}{1 + \epsilon}$ as high as an optimal MPRP-MVS algorithm when run on $S$.
\end{lemma}

\begin{lemma} \cite{Armaselu-arXiv-2020} 
Let $A$ be an algorithm for a MPRP (resp., MPRP-M) and let $P(I)$ be the profit obtained by $A$ when run on an instance $I$ with one vehicle.
Then, on an instance $I'$ with $m$ vehicles, the profit obtained by $A$ is $P(I') \geq \frac{P(I)}{(1 + \frac{1}{1 + \sqrt{m}})^2}$.
\end{lemma}

\hspace{5mm}By running the $\frac{1}{44 \log T}$-optimal MPRP-M algorithm in Section \ref{s:mprp-m} on $I''$, 
we get the following a profit at least $\frac{1}{44 \log T (1 + \frac{1}{1 + \sqrt{m}})}$ as high as an optimal MPRP-MVS algorithm when run on $I$.

\hspace{5mm}Since travel costs are less than the rewards generated by the quantities collected, 
a performance ratio of $\frac{1}{1 + \epsilon}$ in the quantities induces a performance ratio of $\frac{1}{1 + \epsilon}$ in the profits as well.
Putting this together with the lemma above, we get the following.

\begin{lemma}
When run on $I''$, the MPRP-M algorithm in Section \ref{s:mprp-m} may obtain a profit at least $\frac{1}{44 \log T (1 + \epsilon)(1 + \frac{1}{1 + \sqrt{m}})^2}$ as high as an optimal MPRP-MVS algorithm when run on $I$.
\end{lemma}

\hspace{5mm}Note that the reduction described before Lemma 2 takes $O(\frac{n}{\epsilon})$ time.
Then, running MPRP-M on $I''$ which has $|S''| = O(\frac{n}{\epsilon})$ sites, takes $O((\frac{n}{\epsilon})^{11})$ time.
Thus, we have proved the following result.

\begin{theorem}
MPRP-MVS can be solved in $O((\frac{n}{\epsilon})^{11})$ time within an approximation ratio of $44 \log T (1 + \epsilon)(1 + \frac{1}{1 + \sqrt{m}})^2$.
\end{theorem}

\section{\Large{Conclusions and Future Work}}
\label{s:conclusion}

We solve the Multiple Vehicles per Site versions of the Maximum-Profit Routing Problem, specifically,
the Fixed-Supply version and the Time-Variable Supply version.

\hspace{5mm}We leave for future consideration probablisitc approaches for all versions of MPRP, 
e.g. algorithmic solutions whose output is, with high probability, within a certain bound of the optimum for the given instance.
Proving negative results or lower bounds, e.g. inaproximability within a certain ratio, would also be of interest.

\bigskip
\end{document}